\documentclass[letterpaper, 10 pt, conference]{ieeeconf}  

\IEEEoverridecommandlockouts                              

\usepackage{graphics} 
\usepackage{epsfig} 
\usepackage{times} 
\usepackage{amsmath} 
\usepackage{amssymb}  
\usepackage{enumerate} 

\let\labelindent\relax

\usepackage{enumitem}

\usepackage{subfig}
\usepackage{multirow}
\usepackage{mathtools}
\usepackage[english, onelanguage,ruled,vlined, longend, titlenotnumbered]{algorithm2e}

\newcommand{\bfa}{\mathbf{a}}

\newcommand{\hbfa}{\hat{\bfa}}
\newcommand{\bfA}{\mathbf{A}}
\newcommand{\bfb}{\mathbf{b}}

\newcommand{\bfe}{\mathbf{e}}

\newcommand{\bfF}{\mathbf{F}}

\newcommand{\bfg}{\mathbf{g}}

\newcommand{\bfH}{\mathbf{H}}

\newcommand{\bfI}{\mathbf{I}}
\newcommand{\bfJ}{\mathbf{J}}

\newcommand{\bfK}{\mathbf{K}}

\newcommand{\bfL}{\mathbf{L}}

\newcommand{\bfM}{\mathbf{M}}

\newcommand{\bfN}{\mathbf{N}}

\newcommand{\bfn}{\mathbf{n}}

\newcommand{\bfp}{\mathbf{p}}

\newcommand{\bfP}{\mathbf{P}}

\newcommand{\bfQ}{\mathbf{Q}}

\newcommand{\bfR}{\mathbf{R}}

\newcommand{\hbfR}{\hat{\bfR}}

\newcommand{\bfT}{\mathbf{T}}

\newcommand{\bfu}{\mathbf{u}}

\newcommand{\bfv}{\mathbf{v}}

\newcommand{\hbfv}{\hat{\bfv}}
\newcommand{\bfV}{\mathbf{V}}

\newcommand{\bfw}{\mathbf{w}}

\newcommand{\bfx}{\mathbf{x}}

\newcommand{\hbfx}{\hat{\mathbf{x}}}

\newcommand{\bfy}{\mathbf{y}}

\newcommand{\bfzero}{\mathbf{0}}

\newcommand{\bodelta}{\boldsymbol{\delta}}

\newcommand{\bochi}{\protect\raisebox{1pt}{$\boldsymbol{\chi}$}}
\newcommand{\hbochi}{\hat{\bochi}}

\newcommand{\boeta}{\boldsymbol{\eta}}
\newcommand{\boGamma}{\boldsymbol{\Gamma}}

\newcommand{\boomega}{\boldsymbol{\omega}}

\newcommand{\bonu}{\boldsymbol{\nu}}

\newcommand{\boPi}{\boldsymbol{\Pi}}

\newcommand{\boSigma}{\boldsymbol{\Sigma}}

\newcommand{\boUpsilon}{\boldsymbol{\Upsilon}}
\newcommand{\hboUpsilon}{\hat{\boldsymbol{\Upsilon}}}

\newcommand{\boxi}{\boldsymbol{\xi}}

\newcommand{\boXi}{\boldsymbol{\Xi}}

\newtheorem{definition}{Definition}
\newtheorem{theorem}{Theorem}

\newtheorem{prop}{Proposition}

\newcommand{\RR}{{\mathbb R}}

\DeclareMathOperator*{\argmin}{argmin}
\DeclareMathOperator*{\argmax}{argmax}

\usepackage{siunitx}
\usepackage{pgfplots,pgfplotstable}
\usepackage{tikz}
\usetikzlibrary[patterns]
\pgfplotsset{compat=1.12}
\usetikzlibrary{plotmarks}
\usepackage[nospace]{cite}
\usetikzlibrary{arrows,shapes,positioning}
\usetikzlibrary{calc,trees,positioning,arrows,chains,shapes.geometric,%
	decorations.pathreplacing,decorations.pathmorphing,shapes,%
	matrix,shapes.symbols}
\usetikzlibrary{decorations.markings}

\usepackage[top=54pt, bottom=54pt, left=54pt,
right=54pt]{geometry}

\usepackage{setspace}

\usepackage{hyperref}
\hypersetup{hidelinks}

\title{\LARGE \bf Invariant Smoothing with low process noise}

\author{Paul Chauchat$^{2}$, Silv\`ere Bonnabel$^{1}$ and Axel Barrau$^{1,3}$
	\thanks{$^{1}$Centre for Robotics, MINES Paris, PSL Research University,  60 Boulevard Saint-Michel, 75006 Paris, France; and Institut des Sciences Exactes et Appliquees, University of New Caledonia,
		{\tt\small silvere.bonnabel@mines-paristech.fr}}%
	\thanks{$^{2}$IETR, CentraleSupelec, Rennes, France
	{\tt\small paul.chauchat@centralesupelec.fr}}%
	\thanks{$^{3}$OFFROAD, 5 rue Charles de Gaulle, Alfortville, France
		{\tt\small axel@offroad.works }}%
}

\begin{document}
	\maketitle
	\thispagestyle{empty}
	\pagestyle{empty}
	
	\begin{abstract}In this paper we address smoothing -  that is, optimisation-based -  estimation techniques for localisation problems in the case where motion sensors are very accurate. Our mathematical analysis focuses on the difficult limit case where motion sensors are infinitely precise, resulting in the absence of  process noise. Then the formulation degenerates, as the dynamical model that serves as a soft constraint becomes an equality constraint, and conventional smoothing methods are not able to fully respect it. By contrast, once an appropriate Lie group embedding has been found, we prove  theoretically that   invariant smoothing   gracefully accommodates this limit case   in that the estimates tend to be  consistent with the induced  constraints when the noise tends to zero. Simulations on the important problem of initial alignement in inertial navigation show that, in a low noise setting, invariant smoothing may favorably compare  to state-of-the-art smoothers when using precise inertial measurements units (IMU).
	\end{abstract}
	
	\section{Introduction }
	Over the past years,  the smoothing approach has gained ever increasing credit as a state estimator in robotics, owing to the progresses of computers and sparse linear algebra. The rationale is to reduce the consequences of wrong linearisation points \cite{dellaert2006square} through relinarisation. Many of the state-of-the-art algorithms for simultaneous localisation and mapping (SLAM) and visual odometry  are based on smoothing, e.g.,  \cite{kaess2012iSAM2,forster2017SVO2}. It was more recently applied to GPS aided inertial navigation, showing promising results \cite{indelman2013incremental, zhao2014differential, pfeifer2019E-M}. 
	
In parallel, Lie group embeddings have allowed for a new class of filters, see  \cite{hua2014implementation,BayesianLieGroups2011,bourmaud2013discrete}, and in particular the Invariant Extended Kalman Filter (IEKF) \cite{bonnabel2009invariant}, in its modern form \cite{barrau2014invariant}, see \cite{barrau2017annual} for an overview. The IEKF possesses convergence guarantees \cite{barrau2014invariant}, resolves the   inconsistency issues of the EKF for SLAM, see \cite{barrau2015SLAM} and following work \cite{wu2017invariant, caruso2019magneto, heo2018consistent,mahony2017geometric}. For inertial navigation, combining the IEKF with    the Lie group of double spatial direct isometries $SE_2(3)$, or extended poses,  introduced in \cite{barrau2014invariant},  leads to powerful results. In particular, it has led to patented products, see  \cite{barraupatent,barrau2017annual}, and improved legged robot state estimation \cite{hartley2018legged,hartley2018legged2}. Besides their convergence properties as   observers, invariant filters also gracefully accommodate  navigation systems' uncertainty, see \cite{brossard2021associating}. 
Leveraging the framework of Invariant filtering for smoothing, a new estimation algorithm was recently proposed, namely Invariant Smoothing (IS)  \cite{chauchat2018invariant}, see also \cite{walsli2018invariant} and \cite{van2020invariant}.

Another property of the IEKF is that it delivers ``physically consistent'' estimates, when some state variables are known with high degrees of certainty, see  \cite{chauchat2017constraints, barrau2020constraints}.

 In the realm of smoothing algorihtms, low noise  (or equivalently high degrees of certainty) leads to two different kinds of problems:  
 \begin{itemize}
     \item linear matrix inversion problems due to ill-conditioning when solving the linearised problem at each step,
     \item once the linearised problem is properly solved, inconsistent estimates stemming from the nonlinearity of the original problem. 
 \end{itemize}
 The first point is solved in \cite{chauchat2021factor} and won't be considered herein. The second point is the object of the current paper.  
 
 The contributions of this paper are as follows:
 \begin{itemize}
     \item Motivated by the fact that smoothing generally performs better than filtering, we  provide  a theory that consists of the counterpart of the results of \cite{chauchat2017constraints, barrau2020constraints} in the context of smoothing.
     \item IS is shown to better behave than other solvers on a simple wheeled robot localisation example with deterministic dynamics, and the theory gives insight into the reasons why.
     \item The theory is applied to the difficult problem of alignment in inertial navigation systems (INS), i.e.,  IMU-GPS fusion when initial orientation is unknown \cite{wu2013integration, cui2017in_motion},   using a high-grade IMU. Invariant smoothing (IS)  favorably compares to    state-of-the-art  smoothing schemes \cite{forster2016preintegration, gtsam}, as predicted by the theory.
 \end{itemize}

 The superiority of invariant filtering   for alignment, discovered during A. Barrau's thesis \cite{barraupatent,barrauPhD},  has been confirmed in multiple recent works  \cite{fu2021new,cha2021effect,chang2021strapdown,chang2021inertial}, which is the reason why it had first  prompted  patent filing and industrial implementations  \cite{barraupatent}. This has opened avenues for filtering-based alignment, a task generally performed through optimization (for a recent reference  see \cite{ouyang2022optimization}).  However, the optimisation-based invariant approach to alignment has never been explored, as is done in the present paper.

	The paper is organised as follows.  In Section II we apply IS to wheeled robot localization and show in the absence of noise the behavior of IS is   more meaningful than other smoothing algorithms. To explain this feature, we start off by situating  the problem in Section III. Section IV presents the proposed general theory which explains the behavior observed in Section II. In Section V,  the alignment problem in inertial navigation is shown to fit into the proposed framework, using the  Lie group of double direct spatial isometries $SE_2(3)$ \cite{barrau2014invariant}, and the theoretical results are shown to  apply. Low noise simulations show the invariant smoothing approach favorably compares to state-of-the-art smoothers. 
	
	\section{Introductory example}	\label{sec:2d_example}

	Consider  a wheeled mobile robot  in the plane  with unknown initial heading $\theta_0$. The state consists of its  orientation and 2D position $(\theta, \bfx) \in \mathbb S^1 \times \RR^2$.  Let $\mathbf{R}(\theta)\in SO(2)$ denote  the planar rotation of angle $\theta$. For tutorial purposes, assume the robot  \emph{follows a straight line at constant velocity}. This constant velocity motion writes, see e.g., \cite{barrau2014invariant}
	\begin{equation}
	\theta_{i+1} = \theta_i, \quad  
	\bfx_{i+1} = \bfx_i + \mathbf{R}(\theta_i) \mathbf{u}\label{noisefree}
	\end{equation}where 
	$\mathbf{u}=\mathbf{u}_0 dt\in\RR^2$ with $\mathbf{u}_0$ the constant robot's velocity and $dt$ the stepsize. Suppose  that the robot is equipped with differential drives which are perfect, i.e.,   flawlessly reflect the motion is on a straight line (i.e., null angular velocity), and perfectly measure $\mathbf{u}$. Moreover, assume the initial position of the robot $\bar{\bfx}_0\in\RR^2$ is perfectly known. As the initial orientation of the robot (i.e., heading $\theta_0$) is assumed unknown,   the robot's belief about the heading is wrong, see Figure \ref{fig:traj_length2}.  
	If now we receive GPS-based observations of the form $\bfy_k=\bfx_k+\bfn_k$ at some instants $k$, where  $\bfn_k \sim \mathcal{N} (\mathbf{0}, \mathbf{N}_k)$ is a noise that models uncertainty about position measurements, then the robot may calculate the most likely state  trajectory $(\theta_0, \bfx_0), \cdots,(\theta_n, \bfx_n) $  given all  observations up to time $n$. No matter what the observations are, any sensible optimizer should reflect at each step that  the estimated trajectory is a straight line, with known length (as $\mathbf{u}$ is known), but unknown direction $\theta$. 
	
		\begin{figure}[hbt!]
	 \centering	\includegraphics[width=.5 \columnwidth]{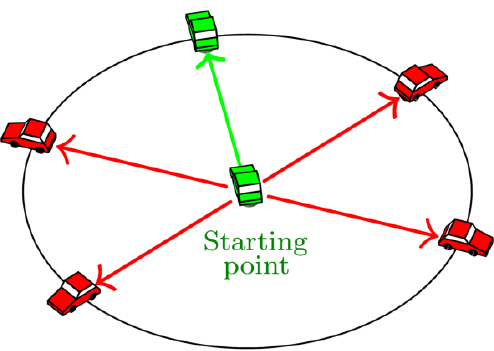}
		\caption{A wheeled robot follows a straight line from a known position with unknown heading. As perfect drives measure the relative displacement, any localization algorithm should ``reflect'' the car is on a circle centered on the initial position with known radius, and optimising over the entire state $(\theta,\bfx)$  to account for GPS  position measurements should boil  down to optimising over heading $\theta$ only.}
		\label{fig:traj_length2}
	\end{figure}
However, a simple numerical experiment where the vehicle moves along a line at a speed of $7 m/s$ with known initial position and a $-3\pi/4$ wrong initial heading (with an initial covariance matrix $diag((3\pi/4)^2,0,0)$) proves this is not the case for standard smoothing methods, see 	Figure \ref{fig:traj_length}. This is because the information about the length is not a hard constraint for the optimisation algorithm. Neither is it  for IS, but the latter's descent step based on the invariant filtering framework \cite{barrau2017annual} inherently respects this information.
	
	\begin{figure}
		\includegraphics{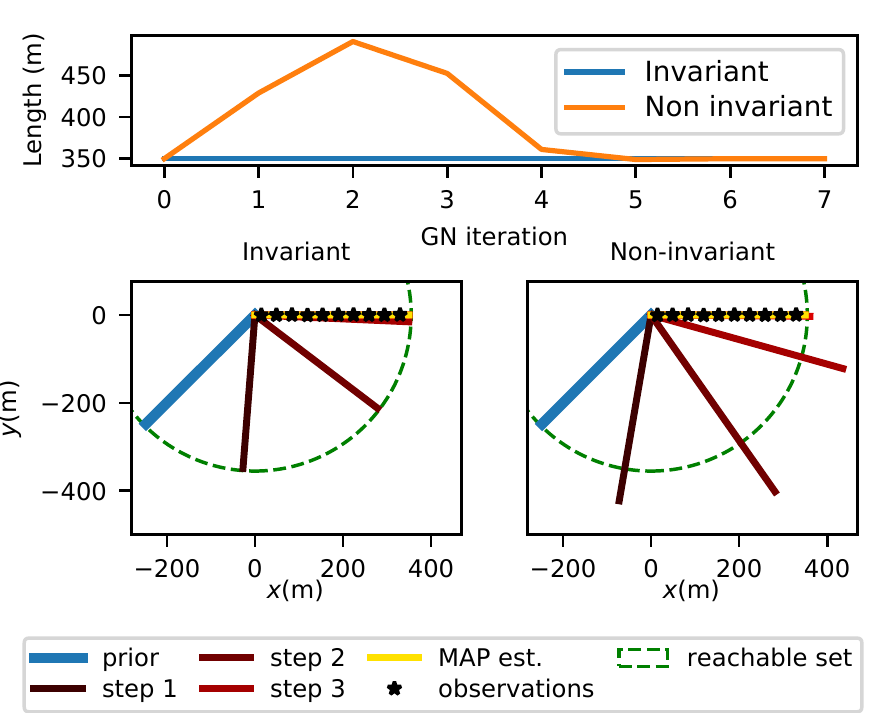}
		\caption{Conventional \cite{dellaert2006square} and Invariant Smoothing  \cite{chauchat2018invariant} of the entire trajectory. Top: Trajectory length of successive iterations for both methods. Bottom: Estimated trajectories from the odometry, at the first descent steps, and at convergence. Both methods maintain a straight line, but only IS keeps a fixed and correct length \emph{at each descent step}, being consistent with the uncertainty in the problem.}
		\label{fig:traj_length}
	\end{figure}
	
	 The remainder of this paper is devoted to the theoretical explanation of what is observed, and to the application of the results to the more challenging problem of inertial alignment. 
	
	\section{Lie group embeddings of the state space}
	\label{sec:lie_group}

	We first briefly recall the invariant filtering framework \cite{barrau2014invariant, barrau2017annual}. Owing to space limitation, we assume the reader has basic knowledge about Lie groups for robotics, and is referred to \cite{barfoot2017robotics} for a general presentation. We consider a state $\bochi \in G$, a matrix Lie group of dimension $q$. Its Lie algebra $\mathfrak{g}$ is identified with $\RR^q$. Thus we consider its exponential map to be defined as $\exp : \RR^q \rightarrow G$. As it is locally invertible, we denote its inverse by $\log$. We recall the notion of adjoint operator matrix of $\bochi \in G$, $\mathbf{Ad}_{\bochi}$, which satisfy
	\begin{equation}
	\forall \bochi \in G, \boxi \in \RR^q,\ \bochi^{-1} \exp(\boxi) \bochi =  \exp(\mathbf{Ad}_{\bochi} \boxi) 
	\label{eq:adjoint}
	\end{equation}
	Group automorphisms are bijective maps $\Phi : G \rightarrow G$ satisfying $\Phi(\bochi \boeta) = \Phi(\bochi) \Phi(\boeta)$ for $\bochi, \boeta \in G$. The Lie group Lie algebra   correspondance, see  \cite{barrau2019linear}, ensures for $\emph{any}$ automorphism $\Phi$ there is $\mathbf{M} \in \RR^{q\times q}$ so that
	\begin{equation}
	\forall (\bochi, \boxi) \in G\times \RR^q,\ \Phi(\bochi \exp(\boxi)) = \Phi(\bochi) \exp(\mathbf{M} \boxi),
	\label{eq:LGLA}
	\end{equation}see also   the \emph{log-linearity} property of \cite{barrau2014invariant}. 
The operator $\bonu \mapsto  \bochi^{-1} \bonu \bochi $ is easily checked to be a group automorphism, and in this case we see from \eqref{eq:adjoint} that   $\mathbf{M} =\mathbf{Ad}_{\bochi}$. We define random variables on Lie groups through the exponential, following  \cite{chirikjian2011stochastic2,bourmaud2013discrete,barfoot2017robotics,brossard2021associating,barrau2017annual}. The probability distribution $\bochi \sim \mathcal{N}_L(\bar{\bochi},\mathbf{P})$ for the random variable $\bochi \in G$ is defined as 
\begin{equation}
	\bochi =  \bar{\bochi} \exp \left(\boxi\right), \text{~} \boxi \sim \mathcal{N}\left(\mathbf{0}, \mathbf{P}\right), \label{eq:left_distrib}
	\end{equation}
	In the following, we   consider a discrete-time  trajectory  denoted as $(\bochi_i)_i$ of the following system
	\begin{subequations}
		\begin{align}
		\bochi_0 &\sim \mathcal{N}_L(\bar{\bochi}, \mathbf{P}_0),  &\bochi_{i+1} = f_i(\bochi_i)
		\label{eq:l-is_prior_noise}\\
		y_k &= h_k(\bochi_{I_k}) + \bfn_k 
		\label{eq:noisy_left_action} & \bfn_k \sim \mathcal{N}(0, \mathbf{N}_k)
		\end{align}
		\label{eq:l-is_system}
	\end{subequations}where $f_i$ is the dynamics function, $\mathbf{P}_0 \in \RR^{q \times q}$ the initial state error covariance, $\mathbf{N}_k \in \RR^{r \times r}$ the observation noise covariance, and $\bochi_{I_k}$ denotes a subset of the states which are involved in the measurements.   Thus \eqref{eq:l-is_system} reflects deterministic dynamics with noisy observations and uncertain initial state. Even if the framework of noise-free dynamics is unrealistic, it allows for a theory that studies how the smoother degenerates when noise tends to zero, as was already done in the context of Kalman filtering in  \cite{chauchat2017constraints, barrau2020constraints}.

	\subsection{Group-affine Dynamics}
	In the invariant framework, $f_i$ is assumed to be group affine. These dynamics were introduced in continuous time in  \cite{barrau2014invariant}, and in discrete time in \cite{barrau2019linear}. The main idea is that they extend the notion of linear dynamics (i.e. defined by affine maps)  from vector spaces to Lie groups.
	\begin{definition}Group affine dynamics are defined through
		\begin{align}
		\bochi_{i+1} = f_i(\bochi_i) = \boGamma_i \Phi(\bochi_{i})\boUpsilon_i. \label{eq:group_affine_def}
		\end{align}where $\boGamma_i,\boUpsilon_i\in G$, and $\Phi$ is an automorphism.
	\end{definition}
	Group affine dynamics include a large class of systems of engineering interest revolving around navigation and robotics, as shown in e.g. \cite{barrau2014invariant,barrau2019linear, walsli2018invariant, mahony2017geometric}. Note that, since $\bonu \mapsto  \bochi^{-1} \bonu \bochi $ is a group automorphism, it is sufficient to define $f_i(\bochi_i) = \Phi(\bochi_{i})\boUpsilon_i$. Both this  and  \eqref{eq:group_affine_def} prove equivalent, but the latter  fits  the equations of inertial navigation  better  \cite{barrau2019linear,brossard2021associating}.

	Group affine dynamics come with the \emph{log-linear property}, originally introduced and proved in \cite{barrau2014invariant} and whose discrete-time counterpart is easier once  \eqref{eq:LGLA} has been identified. 
	\begin{prop}[from \cite{barrau2019linear}, discrete-time log-linear  property]
For group affine dynamics \eqref{eq:group_affine_def}, we have \begin{align}f_i(\bochi_i \exp(\boxi))=\bochi_{i+1}\exp(\mathbf{F}_i\boxi)\label{loglin:prop}\end{align} with $\mathbf{F}_i=\bf\mathbf{Ad}_{\boUpsilon_i^{-1}}\mathbf{M}$ a linear operator, where $\mathbf{M}$ comes from \eqref{eq:LGLA}.
	\end{prop}
	\begin{proof}Focusing on, e.g.,  the first step, we have
	\begin{align}\boGamma_0 \Phi(\bochi_0 \exp(\boxi)) \boUpsilon_0 &\stackrel{\eqref{eq:LGLA}}{=} \boGamma_0 \Phi(\bochi_0) \exp(\mathbf{M} \boxi) \boUpsilon_0 \nonumber\\ 
	&\stackrel{\eqref{eq:adjoint}}{=} \boGamma_0 \Phi(\bochi_0)\boUpsilon_0 \exp(\bf\mathbf{Ad}_{\boUpsilon_0^{-1}}\mathbf{M} \boxi) \nonumber\\
	&=\bochi_1 \exp(  \bf\mathbf{F}_0\boxi).
	\label{eq:group_affine_prop}
	\end{align}
	\end{proof}

	\subsection{Lie group embedding for the introductory example}
	\label{sec:2d_reachability}
	
  We insist that in the invariant filtering approach, Lie group embedding goes well beyond representing  a state variable (e.g., using a rotation matrix to encode the vehicle's orientation). It is more subtle, as various Lie group embeddings exist: Some bring properties and some do not. Back to the simple  introductory example, the state and   dynamics \eqref{noisefree}  should be embedded in  the Lie group of 2D poses, $SE(2)$:
	\begin{equation}
	\bochi = \begin{bmatrix}	\mathbf{R}(\theta) & \bfx \\ 0_{1 \times 2} & 1	\end{bmatrix},
	\quad
	\bochi_{i+1} = \bochi_i\boUpsilon_i , \quad\boUpsilon_i:=\begin{bmatrix}	
	
	\mathbf{Id} & \bfu \\ 0_{1\times 2} & 1	\end{bmatrix}.
	\end{equation}
	The dynamics obviously write as \eqref{eq:group_affine_def} with $ \Phi(\bochi_{i})= \bochi_{i}$.

As initially the position is known to be $\bar{\bfx}_0$ the uncertainty entirely concerns $\theta$, and the initial state necessarily lies in the subpace $\{(\alpha,\bar{\bfx}_0)|\alpha\in\RR\} $ of the state space. In $SE(2)$ this translates into the initial state being of the form $\{\bar{\bochi}_0 \exp(\alpha \boxi_\theta), \alpha \in \RR\}$, where $\boxi_\theta = (0,0,1)^T$. In the formalism of \eqref{eq:left_distrib}, this translates into a rank 1  covariance matrix  $\mathbf{P}_0$ whose range is spanned by  $\boxi_\theta $.

	\subsection{The Property of Reachability}
	The fact that the uncertainty is concentrated on a circle may be explained through the machinery of Lie groups in a more general setting as follows. Assume the initial state lies in a subspace of the state space defined by 
	\begin{align}
	\bar{\bochi}_0 \exp(\sum_{j=1}^p \alpha_j \boeta_j),\quad (\alpha_1,\cdots\alpha_p)\in\RR^p \label{eq:reachability_zero}
	\end{align}
	with $\boeta_1,\cdots,\boeta_p$   known vectors, and $p\leq q=\mathrm{dim}(G)$. The log-linear property, see \eqref{loglin:prop}, shows by induction that at timestep $i$ the state lies within a subspace of the state space of the same form
	\begin{equation}
	\bar{\bochi}_i \exp(\sum_{j=1}^p \alpha_j \widetilde{\mathbf{F}}_i \boeta_j), \qquad \widetilde{\mathbf{F}}_i = \mathbf{F}_{i-1} \ldots \mathbf{F}_0
	\label{eq:reachability_iter}
	\end{equation}
	
	\begin{definition}For an initial state  of the form \eqref{eq:reachability_zero} and noise-free  group affine dynamics \eqref{eq:group_affine_def}, the set of physically reachable states at timestep $i$ is defined as $ \{\bar{\bochi}_i  \exp\bigl(\sum_{j=1}^p \alpha_j \widetilde{\mathbf{F}}_i\boeta_j \bigr)|\alpha_1,\cdots,\alpha_p\in\RR\}$.	
		\label{def:admissible}
	\end{definition}

To embrace the framework of statistics - as smoothing algorithms aim at computing the most likely trajectory - we need to define uncertainty on the state space being consistent with the notion of reachability.

We define an initial belief on the state to be of the form \eqref{eq:left_distrib} where the initial state's covariance $\mathbf{P}_0$ is of rank $p < q$. Denoting by $\boeta_1,\cdots,\boeta_p$ vectors of the Lie algebra that support $\mathbf{P}_0$, the  initial distribution is then supported by a subspace of the form \eqref{eq:reachability_zero}, and any estimator which is consistent with the probabilistic setting should return estimates lying within the set of reachable states.

For technical reasons, see \cite{chauchat2017constraints, barrau2020constraints}, we will systematically assume  the vectors supporting the initial distribution form a Lie subalgebra:   for all $i,j$ the vector $[\boeta_i,\boeta_j]$, the Lie bracket of $\boeta_i, \boeta_j$ \cite{chirikjian2011stochastic2,barfoot2017robotics}, is a linear combination of $\boeta_1,\cdots,\boeta_p$. 

\textbf{Considered problem:} To summarise, what we would like to do is to devise a smoothing algorithm, that is such that when the initial state distribution is of the form \eqref{eq:left_distrib} where the initial state's covariance $\mathbf{P}_0$ is of rank $p < q$, and the  dynamics are noise-free and group-affine \eqref{eq:group_affine_def},   the estimates  $(\hat\bochi_i)_{1\leq i\leq n}$  all lie within the reachable subset   \eqref{eq:reachability_iter}, and this at each (descent) step of the   optimization procedure.

	\section{Main Result}
	\label{sec:l_is_constraints}
	In this section, we prove that Invariant Smoothing (IS) solves the problem above. By contrast standard smoothing algorithms do not, as shown by Figure \ref{fig:traj_length}. 
	\subsection{Smoothing on Lie groups}
	We first briefly recall the Invariant Smoothing (IS) framework introduced in \cite{chauchat2018invariant}. Departing from a system of the form   \eqref{eq:l-is_prior_noise} with observations \eqref{eq:noisy_left_action}, the goal of smoothing is to find \begin{align}(\bochi_i)_i^* = \argmax_{(\bochi_i)_{1\leq i\leq n}} \mathbb{P}((\bochi_i)_i | y_0, \dots, y_n)\label{lik:eq}\end{align}i.e., the maximum a posteriori (MAP) estimate of the trajectory. It is usually found through the Gauss-Newton algorithm. First we devise a cost function associated to Problem  \eqref{lik:eq} as the negative log likelihood
	$$
	\mathcal C =-\log\bigl( \mathbb{P}((\bochi_i)_{1\leq i\leq n} | y_0, \dots, y_n)\bigr) 
	$$
	that we seek to minimize. Given a current guess of the trajectory's states, $(\hat{\bochi}_i)_i$, the cost function $\mathcal C$  is linearised and then the resulting linear problem is solved exactly, yielding a novel estimate, and so on until convergence. Since $\bochi_i$ belongs to a Lie group, linearisation in IS is carried out as
	\begin{equation}
	\forall 1\leq i\leq n\quad \bochi_i = \hat{\bochi}_i \exp(\boxi_i).
	\label{eq:left_inv_error}
	\end{equation}
	where $(\boxi_i)_i$ are the searched parameters that minimize the linearized cost. 
	When considering an invertible prior $\bfP_0$ and noisy dynamics with covariance matrices $\bfQ_i$, IS linearises the cost $\mathcal C$ as  \cite{chauchat2018invariant}
	\begin{align}
	\tilde {\mathcal  C} =&\|\mathbf{p}_0 + \boxi_0\|_{\widetilde{\bfP}_0}^2 
	\label{eq:lie_group_lq}\\
	&+ \sum_i \|\hat{\bfa}_i - \mathbf{F}_i \boxi_i + \boxi_{i+1}\|_{\mathbf{Q}_i}^2 
	+ \sum_k \|\hat{\bfn}_k + \mathbf{H}_k \boXi\|_{\mathbf{N}_k}^2
	\nonumber
	\end{align}
	where we used the notation $\|\mathbf{Z}\|_{\mathbf{P}}^2=\mathbf{Z}^T \mathbf{P}^{-1} \mathbf{Z}$, and where $\boXi$ is the concatenation of $(\boxi_i)_i$. \eqref{eq:lie_group_lq} relies on the Baker-Campbell-Haussdorff formula \cite{barfoot2017robotics} $\log(\exp(a)\exp(b)) = BCH(a,b)$. $\widetilde{\bfP}_0 = \mathbf{J}_0^{-1} \mathbf{P}_0 \mathbf{J}_0^{-T}$, where $\mathbf{J}_0$ is the right Jacobian of the Lie group $G$ \cite{barfoot2017robotics,chirikjian2011stochastic2}, satisfying $BCH(\mathbf{p}_0, \boxi) = \mathbf{p}_0 + \mathbf{J}_0 \boxi + o(\|\boxi\|^2)$, $\mathbf{p}_0 = \log(\bar{\bochi}_0^{-1} \hat{\bochi}_0)$ with a prior $\bar{\bochi}_0$, $\hat{\bfa}_i = \log(f_i(\hbochi_i)^{-1} \hbochi_{i+1})$, $\hat{\bfn}_k = \mathbf{y}_k - h_k(\hbochi_{I_k})$, and $\mathbf{F}_i, \mathbf{H}_k$ are the (Lie group) Jacobians of $f_i$ and $h_k$ respectively. $\mathbf{H}_k$ was padded with zero blocks for the indices not contained in $I_k$.  
	The principle of smoothing algorithms is to solve the linearized problem \eqref{eq:lie_group_lq} in closed form, and to update the trajectory  substituting the optimal $\boxi_i$ in \eqref{eq:left_inv_error}. The problem is then relinearised at this new estimate until convergence. 
	
\subsection{Smoothing with no process noise and degenerate prior}
\label{sec:no_noise}
	
However, in this paper we assumed the dynamics \eqref{eq:l-is_prior_noise} to be noise-free, that is,    $\mathbf{Q}_i = 0$, and $\bfP_0$ to be rank-deficient. As a result, the standard formulation  \eqref{eq:lie_group_lq} appears ill-defined. Moreover, when process noise is low this makes the normal equations solving it ill-conditioned. Theoretically, it turns out that  a) \eqref{eq:lie_group_lq} has a well-defined solution when $\mathbf{Q}_i \to 0$ and b) it is possible to solve  \eqref{eq:lie_group_lq} while avoiding matrix inversions, see 
\cite{chauchat2021factor}. In the present paper, this is none of our concern, and we assume a solver, e.g., \cite{chauchat2021factor}, is able to flawlessly solve \eqref{eq:lie_group_lq} for arbitrarily small process noise, even in the limiting case where $\mathbf{Q}_i \to 0$ and $\mathbf{P}_0  $ is rank-deficient. Our concern is to study the consequences of this limiting case on the state updates. This provides insight in turn into the good behavior of the algorihtm in the presence of low process noise, as occurs in some  applications like inertial navigation.

	\subsection{Main Result}
Assuming \eqref{eq:lie_group_lq} may be properly solved, even in the case of no process noise and rank-deficient $\mathbf{P}_0$, we show now that the batch Invariant Smoother yields estimates which are consistent with the physics of the problem (in other words the assumed uncertainty) at each descent step. 
	\begin{theorem}
		\label{thm:smoothing_constraints}
		Consider the system described by noise-free dynamics \eqref{eq:l-is_prior_noise} assumed to be  group affine.  Let $(\hbochi_i)_i$ represent the current estimates of an Invariant Smoother \cite{chauchat2018invariant}.  Then   \emph{every} iteration of the optimization algorithms exhibits the two following properties (if initalised accordingly):
		\begin{enumerate}[label=(\roman*)]
		    \item  \textbf{Limiting equality constraints}. Equality constraints induced by noise-free dynamics are seamlessly handled by the unconstrained optimization algorithm, which is such that at all steps we have $\hbochi_{i+1} = f_i(\hbochi_i)$.
		    \item \textbf{Belief-compatible estimates}. Assume the prior about the initial state is such that $\mathbf{P}_0$ in \eqref{eq:l-is_prior_noise} is supported by a vector space $\mathbf{V}_0$   of dimension $p < q$,   spanned by, say,  $\boeta_0, \ldots, \boeta_p$, and   such that for all $i,j$, $[\boeta_i,\boeta_j] \in \mathbf{V}_0$: all iterations of the algorithm are in the reachable subspace.
		\end{enumerate}
	\end{theorem}

\begin{proof}
We detail the proof of the theorem for a simplified case, where only two states are considered, i.e. one propagation step.  
	Consider the estimates of a two states trajectory $(\hbochi_0, \hbochi_1)$, where $\hbochi_0$ is reachable, and satisfying  $\hbochi_1 = f_0(\hbochi_0)$. After the next IS update, they will become $(\hbochi_0 \exp(\boxi_0^*), \hbochi_1 \exp(\boxi_1^*))$ where a linear solver returns the solutions $\boxi_0^*,\boxi_1^*$ to \eqref{eq:lie_group_lq} in the considered degenerate case. We want to prove \begin{enumerate}[label=(\roman*)]
	    \item $\hbochi_1 \exp(\boxi_1^*) = f_0(\hbochi_0 \exp(\boxi_0^*))$,
	    \item  if $\hbochi_0,\hbochi_1$ lie in their respective reachable subspaces, so do  $\hbochi_0 \exp(\boxi_0^*), \hbochi_1 \exp(\boxi_1^*)$.
	\end{enumerate}To do so we start proving    IS is such that in the present case
	 \begin{align}
	   	 \boxi_1^* &= \bfF_0 \boxi_0^*\label{prop:fin:eq}\\ \boxi_0^* &\in \bfV_0\label{prop:fin2}
	 \end{align}
	\eqref{prop:fin:eq} implies (i) from the log-linear property \eqref{loglin:prop}. As $\hbochi_0$ is in the reachable subspace, and as $\bfV_0$   forms a Lie subalgebra (hence the technical assumption of stability by Lie bracket), we see \eqref{prop:fin2} implies (ii) as concerns $\boxi_0^*$ and the similar property regarding $\boxi_1^*$ will immediately stem from \eqref{prop:fin:eq}.  As the remainder of the proof is more technical and less insightful, and requires results from \cite{chauchat2021factor}, it has been moved to the appendix.

\end{proof}

 Note, first, that this theorem holds for any solver capable of handling $\bfQ_i = \bfzero$ and rank-deficient $\bfP_0$. Moreover, it is stronger than just saying all states are individually reachable. Here, they all share the same $(\alpha_j)_{1\leq j \leq p}$ from \eqref{eq:reachability_iter}. Finally, it underlies  the results observed in Figure \ref{fig:traj_length}: The fact each iteration appears to be a possible trajectory of the noise-free dynamics  \eqref{noisefree} stems from (i), that is, each trajectory intermediate estimate is a straight line with correct length, by contrast to the the standard smoother that distorts the trajectory at each optimization step. 
 The fact all estimates belong to circles that are compatible with the initial belief encoded in the covariance matrix $diag((3\pi/4)^2,0,0)$ stems from (ii).


	\section{Application to INS Alignment}
	\label{sec:inertial}
	
 In ``genuine'' Inertial Navigation Systems (INS), an initialisation process that relates the body frame to the world frame is required, and this process is called alignment,  see  e.g., \cite{wu2013integration, cui2017in_motion,fu2021new,chang2021strapdown}. This is a challenging process that takes time as the orientation of the carrier is difficult to estimate (the vertical is rapidly found as it is sensed by the accelerometers, but the geographic North is much more difficult to observe). As a result, the main uncertainty during the whole process is dispersed almost exclusively around the vertical axis but it may be very large since the use of magnetometers is generally banned (they are too imprecise and too sensitive to metallic and electromagnetic materials around). Of course alignment is afforded only by highly precise gyrometers, which  justifies the use of a very low noise. 
 
 We consider herein   low-noise and unbiased inertial sensors, to illustrate the practical implications of the noise-free results. We also show the advantage of IS over    state-of-the-art smoothing methods for inertial navigation \cite{forster2016preintegration, gtsam}.
	
	\subsection{Lie Group Embedding}
	
	Important discoveries of \cite{barrau2014invariant} are  the group-affine property and the introduction of $SE_2(3)$ as a Lie group embedding which makes navigation equations group-affine.
	\subsubsection{Unbiased inertial navigation is group affine}
     Consider a robot equipped with an IMU. For unbiased navigation, the state consists of the attitude be $\bfR$, velocity $\bfv$ and position $\bfx$. Unbiased inertial navigation's dynamics are given by 
	\begin{gather}
	f_{\boomega, \bfa}\begin{pmatrix} \bfR \\ \bfv \\ \bfx \end{pmatrix} = 
	\left\{ \begin{matrix}
	\bfR \exp_{SO(3)}(dt (\boomega + \bfw_g)) \\
	\bfv + dt (\bfR (\bfa + \bfw_a) + \bfg) \\
	\bfx + dt\ \bfv
	\end{matrix}
	\right.
	\label{eq:inertial_prop}
	\end{gather}
	with $\boomega,\bfa\in\mathbb{R}^3$  the gyrometers and accelerometers signals  respectively,  $\bfw_g, \bfw_a$ the associated white noises, and $\bfg$ be the gravity vector.

	Following \cite{barrau2014invariant}, the set of navigation triplets $(\bfR, \bfv, \bfx)$ can be endowed with a matrix Lie group structure, called $SE_2(3)$, and referred to as the group of double direct spatial isometries \cite{barrau2014invariant} or extended poses \cite{brossard2021associating}, through
	\begin{align}
	SE_2(3):=\left\{\bfT = \begin{bmatrix}\begin{array}{ccc}\bfR &\bfv&\bfx\\\bfzero_{3\times2} & \multicolumn{2}{c}{\bfI_2} \end{array} \end{bmatrix} \in \RR^{5\times5}\Bigg| \begin{matrix}\bfR \in SO(3)\\\bfv \in \RR^3\\ \bfp \in \RR^3\end{matrix}\right\}.\nonumber
	\end{align}
	In this setting, \eqref{eq:inertial_prop}, defines group affine dynamics (see \cite{brossard2021associating})
	\begin{align}
	\boGamma_i = &\begin{bmatrix}\begin{array}{ccc}\mathbf{Id} & dt \bfg & \bfzero\\\bfzero_{3\times2} & \multicolumn{2}{c}{\bfI_2} \end{array} \end{bmatrix},~
	\Phi(\bfT) = \begin{bmatrix}\begin{array}{ccc}\bfR &\bfv&\bfx + dt \bfv\\\bfzero_{3\times2} & \multicolumn{2}{c}{\bfI_2} \end{array} \end{bmatrix} \nonumber \\
	\boUpsilon_i &= \begin{bmatrix}\begin{array}{ccc}\exp_{SO(3)}(dt \boomega) &dt \bfa_i & \bfzero\\\bfzero_{3\times2} & \multicolumn{2}{c}{\bfI_2} \end{array} \end{bmatrix}
	\end{align}
	
	Let us illustrate how the propagation factors of IS are obtained. Let the residual be $\log(f_i(\bochi_i)^{-1} \bochi_{i+1}) = \log(\mathbf{\Delta_{IMU})}$. The Jacobian is computed with \eqref{eq:group_affine_def} and \eqref{eq:left_inv_error}:
	\begin{align}
	\mathbf{\Delta}_{IMU} & = \boUpsilon_i^{-1} \Phi(\exp(-\boxi_i) \hbochi_i^{-1}) \boGamma_i^{-1} \hbochi_{i+1} \exp(\boxi_{i+1})
	\label{eq:imu_residual}\\
	&= \exp(-\bfF_i \boxi_i) f_i(\hbochi_i)^{-1} \hbochi_{i+1} \exp(\boxi_{i+1}) \\
	\log(\mathbf{\Delta}&_{IMU}) \approx -\bfF_i \boxi_i + \boxi_{i+1} + \log(f_i(\hbochi_i)^{-1} \hbochi_{i+1}),
	\label{eq:inv_imu_jac}
	\end{align}
	where $\bfF_i = \mathbf{Ad}_{\boUpsilon_i^{-1}} \bfM$, which are given on $SE_2(3)$ by
	\begin{equation*}
	{\scriptstyle
		\mathbf{Ad_\bfT} = \begin{bmatrix}
		\bfR & \bfzero_{3\times3} & \bfzero_{3\times3} \\
		\bfv_\times \bfR& \bfR & \bfzero_{3\times3} \\
		\bfp_\times \bfR&  \bfzero_{3\times3} & \bfR
		\end{bmatrix}
		\quad 
		\bfM = \begin{bmatrix}
		\bfI_3 & \bfzero_{3\times 3} & \bfzero_{3\times 3} \\
		\bfzero_{3\times 3} &\bfI_3 & \bfzero_{3\times 3} \\
		\bfzero_{3\times 3} & dt \bfI_3 & \bfI_3
		\end{bmatrix}.
	}
	\end{equation*}

	\subsubsection{Uncertainty propagation} 
	On $SE_2(3)$, the true IMU measurement $\tilde{\boUpsilon}$ can be related to the noisy ones $\boUpsilon$ through $\tilde{\boUpsilon} = \boUpsilon \exp(\bfw)$, where $\bfw$ is a white noise on $\RR^9$ and $\exp$ denotes the exponential map of $SE_2(3)$.   For more on $SE_2(3)$, and its use for inertial navigation (notably the derivation of the covariance process noise matrix) the reader is referred to  \cite{barrau2014invariant, brossard2021associating}. 
 
	\subsection{Difference between IS and other Smoothers}
	Let us compare IS with the state-of-the-art smoothing methods  for inertial data \cite{forster2016preintegration},   and the one implemented in GTSAM \cite{gtsam} (which slightly differs). The considered residuals are essentially the same, and so are their covariances, although obtained through less tedious computations.
	The main difference lies in the parametrisation of the state (i.e. the retraction) used to update the state variables at each optimization descent step. 
	Indeed, the retractions used in \cite{forster2016preintegration} and GTSAM \cite{gtsam} are respectively
	\begin{align}
	(\hbfR, \hbfv, \hbfx) &\leftarrow (\hbfR \bodelta_R, \hbfv + \bodelta_v, \hbfx + \hbfR \bodelta_x),
	\label{eq:forster_retraction}\\
	(\hbfR, \hbfv, \hbfx) &\leftarrow (\hbfR \bodelta_R, \hbfv+\hbfR \bodelta_v, \hbfx + \hbfR \bodelta_x).
	\label{eq:gtsam_retraction}
	\end{align}
	which are linear by nature whereas the exponential map offers a fully nonlinear appropriate map. Note that \eqref{eq:gtsam_retraction} is a first-order approximation of the Lie exponential on $SE_2(3)$.  Jacobians for \eqref{eq:gtsam_retraction} can be retrived from \eqref{eq:imu_residual}, as
	\begin{align}
	\mathbf{\Delta}_{IMU} 
	&= f_i(\hbochi_i)^{-1} \hbochi_{i+1} \exp(-\hat{\bfF}_i \boxi_i)  \exp(\boxi_{i+1}),
	\label{eq:for_imu_jac}
	\end{align}
	where $\hat{\bfF}_i = \mathbf{Ad}_{\hboUpsilon_i^{-1}} \bfM$ is the wanted Jacobian, with $\hboUpsilon_i = \Phi(\hbochi_i)^{-1} \boGamma_i^{-1} \hbochi_{i+1}$ representing the ``estimated'' measurement. Jacobian for \eqref{eq:forster_retraction} can then be easily derived. The other difference is that IS uses the logartihm map of $SE_2(3)$.

	\begin{figure}[h]
	\centering
	\includegraphics[width=.95\linewidth]{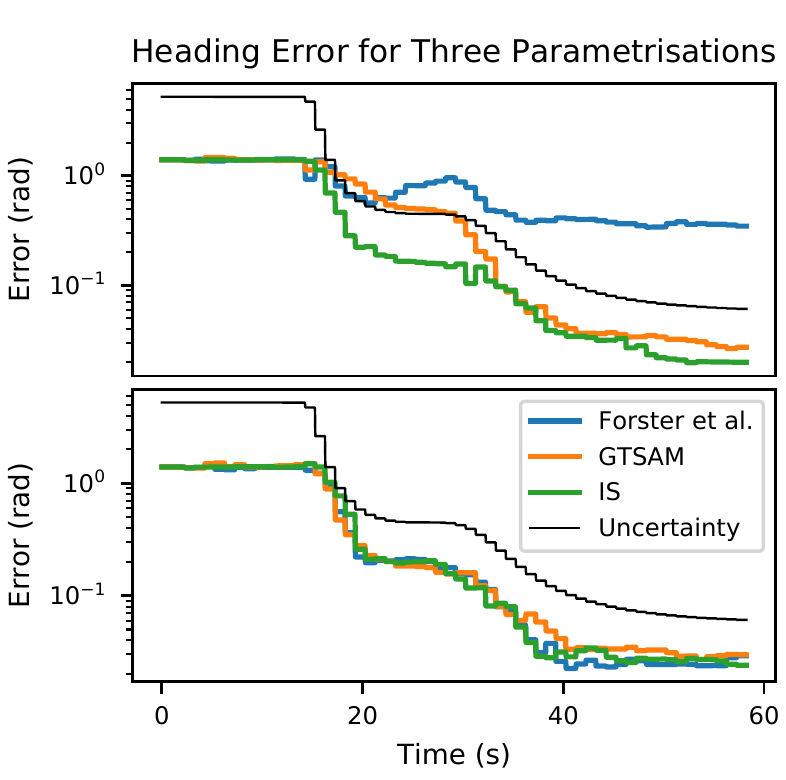}
	\caption{Yaw error  (on a log scale) for the alignment problem of Section~\ref{sec:inertial}, starting with initial heading error of 80$^\circ$ over time. IS is compared with \cite{forster2016preintegration} and GTSAM \cite{gtsam}. ``Uncertainty'' denotes the $3\sigma$ envelope of the IS estimate. Top: sliding window of size 10. Bottom: sliding window of size 50. }
	\label{fig:yaw_error}
	\end{figure}

	\subsection{Experimental Setting}
	\label{sec:experiments}
	We compare the three smoothing methods on a simulated in-motion alignment problem. A vehicle is equipped with a precise IMU and a GPS sensor. The IMU and GPS measurements are acquired at 200 Hz and 1 Hz respectively, and considered with the following standard deviations
	\begin{equation}
	\sigma_g = 2.7e-4~ ^\circ/s , ~ \sigma_a = 1.5e-3 m/s^2, ~\sigma_n = 3 m
	\label{eq:IMU_prop_sigma}
	\end{equation}
	The initial position is supposed to be   known, as is customary for initial alignment, but with unknown speed and attitude:
	\begin{equation}
		\sigma_p^0 = 0 m, \sigma_v^0 = 10 m/s, \sigma_{R}^0 = 100~ ^\circ
	\label{eq:init_sigma}
	\end{equation}
	The trajectory starts with the vehicle standing still for 15s, before starting to move forward for 25s. The estimate is initialised with zero velocity, correct roll and pitch, and an incorrect heading of 80$^\circ$, as it may be assumed that roll and pitch are rapidly identified, as they are highly observable.
	The IMU is preintegrated between each GPS measurements, see \cite{forster2016preintegration, barrau2019preintegration}, where updates occur. The estimation was carried out in a sliding window setting, where the oldest state is marginalised out once the maximum of states is reached. Two experiments were carried out, with windows of size 10 and 50, so that in the first one marginalisation starts before the yaw has converged. One Gauss-Newton iteration is carried out at each update.

	\subsection{Results}

	The results are displayed on Figure \ref{fig:yaw_error}. Although the whole navigation state is estimated, only the yaw error is reported, as it is the key parameter which is difficult to estimate. The RMSE is computed over 10 Monte Carlo runs. The $3\sigma$ bound of the yaw estimate of IS is also reported (the bounds of both other methods are very similar). In the top chart, which involves a sliding window of 10 time steps, we see that  \cite{forster2016preintegration} (Forster et al.)  becomes inconsistent due to early marginalisation. As concerns the two other algorithms, IS and GTSAM \cite{gtsam} both coincide after convergence indeed, but IS shows quicker convergence and better consistency since GTSAM exceeds the $3\sigma$ bound between 20 and 30 seconds. This is due to the fact that GTSAM uses $\hat{\boUpsilon}_i$, which becomes erroneous after update \eqref{eq:gtsam_retraction}, even with $\bfQ_i = \bfzero$. In the case of a sliding window of size 50 (bottom chart), \cite{forster2016preintegration} converges to the IS and GTSAM estimates, since they share the same cost function. Indeed, the vehicle starts moving before marginalisation occurs, so less errors are propagated, ensuring better estimators' consistency.

	\section{Conclusion}
	This paper first presented a new theoretical property of the recently introduced Invariant Smoothing (IS) framework, which was shown to respect a class of geometrical constraints appearing in the limit-case of noise-free dynamics, advocating for its use in high-accuracy navigation. This was illustrated by a 2D introductory wheeled robot localisation simulated problem, for which only IS managed to produce consistent successive iterations. The impact of this result for unbiased inertial navigation, with low but non-null process noise, was then evaluated on alignment simulations using a high-grade IMU. In this case, IS proved more stable and consistent than state-of-the-art inertial smoothing methods. Future work will further study the impact of the window size on smoothing methods, and how this adapts to biased inertial navigation, using the recently introduced two-frames group  \cite{barrau2022geometry} providing a novel embedding that better accommodates sensor biases.

	\appendix

We now complete the proof of the theorem. We first recall results of  \cite{chauchat2021factor}. As concerns \eqref{eq:lie_group_lq}, it may be re-written as   
	\begin{align}
	\begin{bmatrix}
	\boxi_0^* \\ \boxi_1^*
	\end{bmatrix} = \argmin_{\boxi_0, \boxi_1} 
	\left\|\bfA_0
	\begin{bmatrix}
	\boxi_0 \\ \boxi_1
	\end{bmatrix} - \bfb_0 \right\|_{\boPi_0}^2 
	+ \left\|\bfH_1
	\begin{bmatrix}
	\boxi_0 \\ \boxi_1
	\end{bmatrix} - \hat{\bfn}_1 \right\|_{\bfN_1}^2,
	\label{eq:partition}
	\end{align}
	where $\bfA_0 = \begin{bmatrix}
	\mathbf{Id} \\ -\bfF_0 & \mathbf{Id}
	\end{bmatrix}$, $\bfb_0 = \begin{bmatrix} \bfp_0 \\ \hbfa_0 \end{bmatrix}$, $\boPi_0 = \text{diag}(\widetilde{\bfP}_0, \bfQ_0)$ defined in \eqref{eq:lie_group_lq}. This yields  a solution to \eqref{eq:partition} as (see \cite{chauchat2021factor})
	\begin{subequations}
		\begin{align}
		\begin{bmatrix}
		\boxi_0^* \\ \boxi_1^*
		\end{bmatrix} &= \mathbf{A}_0^{-1}
		\left( \left( \mathbf{Id} - \mathbf{K}\mathbf{L} \right) \bfb_0
		+ \mathbf{K} \hat{\bfn}_1 \right)\label{eq:linear_lq_solution2}
		\\
		\mathbf{L} = \bfH_1 \bfA_0^{-1} &\qquad \mathbf{K} = \boPi_0 \mathbf{L}^T (\mathbf{L} \boPi_0 \mathbf{L}^T + \bfN_1)^{-1}	\label{eq:linear_lq_solution}
		\end{align}
		\label{eq:robust_linear_lq_solution}
	\end{subequations}
	By assumption, $\bfQ_0 = \bfzero$ and $\hbfa_0 = \log(f_0(\hbochi_0)^{-1} \hbochi_1) = \bfzero$. Note that \eqref{eq:robust_linear_lq_solution} provides a solver accomodating $\bfQ_0 = \bfzero$ and rank-deficient $\bfP_0$, as mentioned in Section~\ref{sec:no_noise}. Let $\bfL = \begin{bmatrix} \bfL_0 & \bfL_1 \end{bmatrix}$, we then have
  
	$$ \bfK = \begin{bmatrix} \widetilde{\mathbf{P}}_0 \\ & \bfzero \end{bmatrix} \begin{bmatrix} \bfL_0^T \\ \bfL_1^T \end{bmatrix} \underbrace{(\mathbf{L} \boPi_0 \mathbf{L}^T + \bfN_1)^{-1}}_{\boSigma} = \begin{bmatrix}  \widetilde{\bfP}_0 \bfL_0^T \boSigma \\ \bfzero \end{bmatrix} $$
	\begin{equation}
	    \begin{bmatrix}
	    \boxi_0^* \\ \boxi_1^*
	    \end{bmatrix} = \begin{bmatrix} \mathbf{Id} \\ \bfF_0 & \mathbf{Id} \end{bmatrix} \begin{bmatrix} \bfp_0 + \widetilde{\bfP}_0 \bfL_0^T  \boSigma (\hat{\bfn}_1 - \bfL_0 \bfp_0) \\ \bfzero \end{bmatrix}
	\end{equation}
	Recall that $\widetilde{\bfP}_0 = \mathbf{J}_0^{-1} \mathbf{P}_0 \mathbf{J}_0^{-T}$ from \eqref{eq:lie_group_lq}. By assumption, $\bfp_0 = \log(\bar{\bochi}_0^{-1} \hbochi_0) \in \bfV_0$ and $\bfV_0$ is a Lie subalgebra, so for any $\bfe \in \bfV_0$, $\bfJ_0^{-1} \bfe \in \bfV_0$ \cite{barfoot2017robotics}. Since $\bfP_0$ is spanned by $\boeta_1, \ldots, \boeta_p \in \bfV_0$, then it has its image in $\bfV_0$, and so does $\widetilde{\bfP}_0$. $\bfV_0$ being closed by addition, this shows that $\boxi_0^* \in \bfV_0$. Moreover, it is straightforward that $\boxi_1^* = \bfF_0 \boxi_0^*$. For longer trajectories, the proof easily generalises, as the involved matrices keep the same structure: $\bfK$ and $\bfb_0$ only have a non-zero first block row, and the first column of $\bfA_0^{-1}$ contains the $\widetilde{\mathbf{F}}_i$ from \eqref{eq:reachability_iter}.

	\bibliographystyle{plain}
	\bibliography{biblio}             

\end{document}